%% file: CS_Gen_LB_v11.tex
\documentclass[final]{IEEEtran}
\usepackage{xr}
\usepackage{cite}
\usepackage{amsmath}
\usepackage[colorlinks=true,linkcolor=black,citecolor=black,urlcolor=black]{hyperref}
\usepackage{cases}
\usepackage[utf8]{inputenc}
\usepackage[english]{babel}
\usepackage{footnote}
\usepackage{booktabs}

\input{preamble}

\begin{document}
    
\title{Information-Theoretic Lower Bounds for Compressive Sensing with Generative Models}

\author{Zhaoqiang Liu and Jonathan Scarlett

\thanks{The authors are with the  Department of Computer Science, School of Computing, National University of Singapore (NUS). Jonathan Scarlett is also with the Department of Mathematics, NUS. Emails: \url{dcslizha@nus.edu.sg};  \url{scarlett@comp.nus.edu.sg}

This work was supported by the Singapore National Research Foundation (NRF) under grant number R-252-000-A74-281.}}

\maketitle

\begin{abstract}
    Recently, it has been shown that for compressive sensing, significantly fewer measurements may be required if the sparsity assumption is replaced by the assumption the unknown vector lies near the range of a suitably-chosen generative model.  In particular, in (Bora {\em et al.}, 2017) it was shown roughly $O(k\log L)$ random Gaussian measurements suffice for accurate recovery when the generative model is an $L$-Lipschitz function with bounded $k$-dimensional inputs, and $O(kd \log w)$ measurements suffice when the generative model is a $k$-input ReLU network with depth $d$ and width $w$.  In this paper, we establish corresponding algorithm-independent lower bounds on the sample complexity using tools from minimax statistical analysis.  In accordance with the above upper bounds, our results are summarized as follows: (i) We construct an $L$-Lipschitz generative model capable of generating group-sparse signals, and show that the resulting necessary number of measurements is $\Omega(k \log L)$; (ii) Using similar ideas, we construct ReLU networks with high depth and/or high depth for which the necessary number of measurements scales as $\Omega\big( kd \frac{\log w}{\log n}\big)$ (with output dimension $n$), and in some cases $\Omega(kd \log w)$.  As a result, we establish that the scaling laws derived in (Bora {\em et al.}, 2017) are optimal or near-optimal in the absence of further assumptions.
\end{abstract}

\section{Introduction}

Over the past 1--2 decades, tremendous research effort has been placed on theoretical and algorithmic studies of high-dimensional linear inverse problems \cite{Fou13,wainwright2019high}.  The prevailing approach has been to model low-dimensional structure via assumptions such as sparsity or low rankness, and numerous algorithmic approaches have been shown to be successful, including convex relaxations \cite{Tro06,Wai09a}, greedy methods \cite{Tro04,Wen2016}, and more.  The problem of sparse estimation via linear measurements (commonly referred to as {\em compressive sensing}) is particularly well-understood, with theoretical developments including sharp performance bounds for both practical algorithms \cite{Wai09a,Don13,Ame14,Wen2016} and (potentially intractable) information-theoretically optimal algorithms \cite{Wai09,Ari13,Can13,Sca15}.

Following the tremendous success of deep generative models in a variety of applications \cite{Fos19}, a new perspective on compressive sensing was recently introduced, in which the sparsity assumption is replaced by the assumption of the underlying signal being well-modeled by a generative model (typically corresponding to a deep neural network) \cite{Bor17}.  This approach was seen to exhibit impressive performance in experiments, with reductions in the number of measurements by factors such as $5$ to $10$ compared to sparsity-based methods.

In addition, the authors of \cite{Bor17} provided theoretical guarantees on their proposed algorithm, essentially showing that an $L$-Lipschitz generative model with bounded $k$-dimensional inputs leads to reliable recovery with $m = O(k \log L)$ random Gaussian measurements (see Section \ref{sec:overview} for a precise statement).  Moreover, for a ReLU network (see Appendix \ref{sec:relu_defs} for brief definitions) generative model from $\bbR^k$ to $\bbR^n$ with width $w$ and depth $d$, it suffices to have $m = O(kd \log w)$.  Some further related works are outlined below.

In this paper, we address a prominent gap in the existing literature by establishing {\em algorithm-independent lower bounds} on the number of measurements needed.
Using tools from minimax statistical analysis, we show that for generative models satisfying the assumptions of \cite{Bor17}, the above-mentioned dependencies $m = O(k \log L)$ and $m = O(kd \log w)$ are tight or near-tight in the absence of further assumptions.  Our argument is based on a reduction to compressive sensing with a {\em group sparsity} model (e.g., see \cite{Bar10}), i.e., forming a generative model that is capable of producing such signals.

\subsection{Related Work}

The above-mentioned work of Bora {\em et al.} \cite{Bor17} performed the theoretical analysis assuming that one can find an input to the generative model minimizing an empirical loss function up to a given additive error (see Theorem \ref{thm:bora1} below).  The analysis is based on showing that Gaussian random matrices satisfy a natural counterpart to the Restricted Eigenvalue Condition (REC) termed the {\em Set-Restricted Eigenvalue Condition} (S-REC).
In practice, minimizing the empirical loss function may be hard, so it was proposed to use gradient descent in the latent space (i.e., the input space of the generative model). 

A variety of follow-up works provided additional theoretical guarantees for compressive sensing with generative models. For example, instead of using a direct gradient descent algorithm as in~\cite{Bor17}, the authors of~\cite{Sha18,jalali2019using} provide recovery guarantees for a projected gradient descent algorithm under various assumptions, where the gradient steps are instead taken in the ambient space (i.e., the output space of the generative model)

In~\cite{Bor17}, the recovered signal is assumed to lie in (or be close to) the range of the generative model $G(\cdot)$, which poses limitations for cases that the true signal is further from the range of $G(\cdot)$. To overcome this problem, a more general model allowing sparse deviations from the range of $G(\cdot)$ was introduced and analyzed in~\cite{Dha18}. 

In the case that the generative model is a ReLU network, under certain assumptions on the layer sizes and randomness assumptions on the weights, the authors of~\cite{Han18} show that the non-convex objective function of empirical risk minimization does not have any spurious stationary points, and accordingly establish a simple gradient-based algorithm that is guaranteed to find the global minimum.

Sample complexity upper bounds are presented in \cite{Wei19} for various generalizations of the original model in \cite{Bor17}, namely, non-Gaussian measurements, non-linear observations models, and heavy-tailed noise.  Another line of works considers compressive sensing via {\em untrained} neural network models \cite{Van18,Hec18}, but these are less relevant to the present paper.

Despite this progress, to the best of our knowledge, there is no previous literature providing {\em algorithm-independent lower bounds} on the number of measurements needed,\footnote{See Section \ref{sec:contr} for the discussion of a concurrent work \cite{Kam19}.} and without these, it is unclear to what extent the existing results can be improved.
On the other hand, algorithm-independent lower bounds are well-established for sparsity-based compressive sensing \cite{Can13,Duc13,Wai09,Aer10,Akc10,Ree13,Aks17,Sca15}, and parts of our analysis will build on the techniques in these works.

\subsection{Contributions} \label{sec:contr}

In this paper, we establish information-theoretic lower bounds that certify the optimality or near-optimality of the above-mentioned upper bounds from \cite{Bor17}.  More specifically, our main results are summarized as follows:\footnote{Throughout the paper, we use the standard asymptotic notation $O$, $o$, $\Omega$, $\omega$, and $\Theta$.}
\begin{itemize}
    \item In Section \ref{sec:lower}, we construct a bounded $k$-input $L$-Lipschitz generative model capable of generating group-sparse signals, and show that the resulting necessary number of measurements for accurate recovery is $\Omega(k \log L)$. 
    \item In Section \ref{sec:relu}, building on the preceding ideas, we construct ReLU networks whose width $w$ and/or depth $d$ are large, requiring $\Omega\big( kd \frac{\log w}{\log n}\big)$ measurements (where $n$ is the output dimension), and in some cases requiring $\Omega(kd \log w)$ measurements.
\end{itemize}
Note that these results are only summarized informally here; see the relevant sections for formal statements, in particular Theorems \ref{thm:main}, \ref{thm:main2}, and \ref{thm:wide_deep}.

In concurrent work \cite{Kam19}, a closely-related result to ours was proved for the bounded Lipschitz setting.  There are some slight differences, in that we consider noisy measurements for signals with no representation error, whereas \cite{Kam19} considers signals with representation error but no output noise.\footnote{These two settings can be brought closer together by replacing our output noise model $\by = \bA\bx + \bz$ by an input noise model of the form  $\by = \bA(\bx + \bz)$.  For Gaussian noise, the two are in fact equivalent up to re-scaling of the noise variance when $\bA$ has orthogonal rows of equal $\ell_2$-norm, and it was argued in \cite{Kam19} that such a restriction on $\bA$ in fact bears no loss of generality.}  The results are, however, closely related, and both can be interpreted as proving the tightness of \cite{Bor17}.  Interestingly, the analysis techniques are quite different and complementary to each other, with \cite{Kam19} using a reduction based on communication complexity.  To our knowledge, our lower bound for deep and/or wide ReLU networks is unique to the present paper.

\section{Problem Setup and Overview of Upper Bounds} \label{sec:overview}

In this section, we formally introduce the problem, and overview one of the main results of \cite{Bor17} giving an upper bound on the sample complexity for Lipschitz-continuous generative models, so as to set the stage for our algorithm-independent lower bounds.  

Compressive sensing aims to reconstruct an unknown vector $\bx^*$ from a number of noisy linear measurements of the form $\by = \bA \bx^* + \bmeta$ (formally defined below).  In~\cite{Bor17}, instead of making use of the common assumption that $\bx^*$ is $k$-sparse \cite{Fou13}, the authors assume that $\bx^*$ is close to some vector in the range of a generative function $G(\cdot)$.  We adopt the same setup as that in \cite{Bor17}, but for convenience we consider both rectangular and spherical input domains (whereas \cite{Bor17} focused on the latter).  In more detail, the setup is as follows:
\begin{itemize}
    \item A {\em generative model} is a function $G \,:\, \calD\to \bbR^n$, with latent dimension $k$, ambient dimension $n$, and input domain $\calD \subseteq \bbR^k$.
    \item When the signal to be estimated is $\bx^* \in \bbR^n$, the observed vector is given by
    \begin{equation}
        \by = \bA \bx^* + \bmeta,
    \end{equation}
    where $\bA \in \bbR^{m\times n}$ is the measurement matrix, $\bmeta \in \bbR^m$ is the noise vector, and $m$ is the number of measurements.  For now $\bx^*$ is arbitrary, but should be thought of as being close to $G(\bz)$ for some $\bz \in \calD$.
    \item We define the $\ell_2$-ball $B_2^k(r):=\{\bz \in \bbR^k: \|\bz\|_2 \le r\}$, and the $\ell_{\infty}$-ball $B_{\infty}^k(r):=\{\bz \in \bbR^k: \|\bz\|_{\infty} \le r\}$.  We will focus primarily on the case that the input domain $\calD$ is of one of these two types, and we refer to these cases as {\em spherical domains} and {\em rectangular domains}. 
    \item For $S \subseteq \calD \subseteq \bbR^k$, we define 
    \begin{equation}
        G(S) = \{ \bx \in \bbR^n \,:\, \bx = G(\bz) \text{ for some }\bz \in S \}.
    \end{equation}
    When $S = \calD$ is the domain of $G$, we also use the notation ${\rm Range}(G) = G(\calD)$, which we call the {\em range} of the generative model.
\end{itemize}

One of the two main results in \cite{Bor17} is the following, providing general recovery guarantees for compressive sensing with generative models and Gaussian measurements.

\begin{theorem} {\em (Upper Bound for Spherical Domains \cite[Thm.~1.2]{Bor17})} \label{thm:bora1}
    Fix $r > 0$, let $G \,:\, B_2^k(r) \rightarrow \bbR^n$ be an $L$-Lipschitz function, and let $\bA \in \bbR^{m \times n}$ be a random measurement matrix whose entries are i.i.d.~with distribution $A_{ij} \sim \calN(0,\frac{1}{m})$.  Given the observed vector $\by = \bA\bx^* + \bmeta$  for some $\bx^* \in \bbR^n$, let $\hat{\bz}$ minimize $\|\by - \bA G(\bz)\|_2$ to within additive error $\epsilon$ of the optimum over $\bz \in B_2^k(r)$. Then, for any number of measurements satisfying $m \ge C k\log \frac{Lr}{\delta}$ for a universal constant $C$ and any $\delta > 0$, the following holds with probability $1-e^{-\Omega(m)}$:
    \begin{equation}\label{eq:upper_bds}
        \|G(\hat{\bz})-\bx^*\|_2 \le 6 \min_{\bz^* \in B_2^k(r)} \|G(\bz^*)-\bx^*\|_2 + 3\|\bmeta\|_2 + 2\epsilon + 2\delta.
    \end{equation}
\end{theorem}

The sample complexity $m=O(k\log \frac{Lr}{\delta})$ comes from the covering number of $B_2^k(r)$ with parameter $\frac{\delta}{L}$ (see Appendix \ref{sec:covering} for formal definitions), which is $\exp\big( O(k\log \frac{Lr}{\delta}) \big)$.  The analysis of \cite{Bor17} extends directly to other compact domains $\calD\subseteq\bbR^k$; in the following, we provide the extension to the rectangular domain $\calD = B_{\infty}^k(r)$, which will be particularly relevant in this paper.
The covering number of  $B_{\infty}^k(r)$ with parameter $\frac{\delta}{L}$ is $\exp\big( O(k\log \frac{Lr\sqrt{k}}{\delta} ) \big)$ (see Appendix \ref{sec:covering}), and as a result, the sample complexity becomes $m=O\big(k\log \frac{Lr\sqrt{k}}{\delta}\big)$, and we have the following.


\begin{theorem} {\em (Upper Bound for Rectangular Domains; Adapted from \cite[Thm.~1.2]{Bor17})} \label{thm:bora2}
    Fix $r > 0$, let $G: B_{\infty}^k(r) \rightarrow \bbR^n$ be an $L$-Lipschitz function, and let $\bA \in \bbR^{m \times n}$ be a random measurement matrix whose entries are i.i.d.~with distribution $A_{ij} \sim \calN(0,\frac{1}{m})$.  Given the observed vector $\by = \bA\bx^* + \bmeta$ for some $\bx^* \in \bbR^n$, let $\hat{\bz}$ minimize $\|\by - \bA G(\bz)\|_2$ to within additive error $\epsilon$ of the optimum over $\bz \in B_{\infty}^k(r)$. Then, with a number of measurements satisfying $m \ge C k\log \frac{Lr \sqrt{k}}{\delta}$ for a universal constant $C$  and any $\delta > 0$, the following holds with probability $1-e^{-\Omega(m)}$:
    \begin{equation}\label{eq:upper_bds2}
    \|G(\hat{\bz})-\bx^*\|_2 \le 6 \min_{\bz^* \in B_{\infty}^k(r)} \|G(\bz^*)-\bx^*\|_2 + 3\|\bmeta\|_2 + 2\epsilon + 2\delta.
    \end{equation}
\end{theorem}

Another main result of \cite{Bor17} (namely, Theorem 1.1 therein) concerns the sample complexity for generative models formed by neural networks with ReLU activations.  We also establish corresponding lower bounds for such settings, but the formal statements are deferred to Section \ref{sec:relu}.

For the sake of comparison with our lower bounds, it will be useful to manipulate Theorems \ref{thm:bora1} and \ref{thm:bora2} into a different form.  To do this, we specialize the setting as follows:
\begin{itemize}
    \item We assume that $\bx^* \in {\rm Range}(G)$, so that the first terms in \eqref{eq:upper_bds} and \eqref{eq:upper_bds2} become zero.  That is, representation error does not play a role in this paper.
    \item We consider the case that the constrained minimum of $\|\by - \bA G(\bz)\|_2$ can be found exactly, so that $\epsilon = 0$.  That is, optimization error also does not play a role.
    \item We consider the case of i.i.d.~Gaussian noise: $\bmeta\sim \calN(0,\frac{\alpha}{m} \bI_m)$, where $\bI_m$ is the $m \times m$ identity matrix, and $\alpha > 0$ is a positive constant indicating the noise level.  By standard $\chi^2$-concentration \cite[Sec.~2.1]{wainwright2019high}, we have $\|\bmeta\|_2^2 \le 2\alpha$ with probability $1 - e^{-\Omega(m)}$, and when this holds, the term $3\|\bmeta\|_2$ is upper bounded by $3\sqrt{2\alpha}$.
    \item We focus on the case that the goal is to bring the error bound down to the noise level, and hence, we set $\delta = \sqrt{\alpha}$ to match the $\|\bmeta\|_2$ term (to within a constant factor) with high probability.
\end{itemize}
In addition, we move from a randomized measurement matrix $\bA$ to a fixed deterministic choice, as our lower bounds will be stated in the latter form.  In this specialized setting, we have the following.

\begin{corollary} \label{cor:bora}
    {\em (Simplified Upper Bounds)}
    Consider the setup of Theorem \ref{thm:bora1} with fixed $\bA$ in place of Gaussian measurements, and with $\bx^* \in {\rm Range}(G)$, no optimization error ($\epsilon = 0$), i.i.d.~Gaussian noise  $\bmeta\sim \calN(0,\frac{\alpha}{m} \bI_m)$ with $\alpha > 0$, and $\delta = \sqrt{\alpha}$.  Then, when $m \ge C k\log \frac{Lr}{\sqrt{\alpha}}$ for a universal constant $C$, there exists some $\bA \in \bbR^{m \times n}$ with $\|\bA\|_{\rmF}^2 = n$ such that
    \begin{equation}
        \sup_{ \bx^* \in {\rm Range}(G) } \bbE\big[ \|G(\hat{\bz}) - \bx^*\|_2^2 \big] \le C' \alpha \label{eq:cor_guarantee}
    \end{equation}
    for a universal constant $C'$.  Analogously, we have the same result under the setup of Theorem \ref{thm:bora2} (i.e., $\calD = B_{\infty}^k(r)$) when $m \ge C k\log \frac{Lr \sqrt{k}}{\sqrt{\alpha}}$.
\end{corollary}
\begin{proof}
    We first prove the case of a spherical domain based on Theorem \ref{thm:bora1} and its proof in \cite{Bor17}.  An inspection of that proof reveals that there are two high-probability events leading to the result holding with probability $1-e^{-\Omega(m)}$:
    \begin{itemize}
        \item $\bA$ satisfies the {\em set-restricted eigenvalue condition} (S-REC) with parameters $(\gamma,\delta)$ and set $S = {\rm Range}(G)$.  This condition states that $\|\bA(\bx_1 - \bx_2)\|_2 \ge \gamma \|\bx_1 - \bx_2\|_2 - \delta$, and in \cite{Bor17} the choice $\gamma = \frac{1}{2}$ is made.
        \item $\bA$ satisfies $\|\bA(\bx^* - \bar{\bx})\|_2 \le 2\|\bx^* - \bar{\bx}\|_2$, where $\bar{\bx} = \argmin_{\bx \in {\rm Range}(G)} \|\bx^* - \bx\|$.
    \end{itemize}
    The second of these events can be ignored in the present simplified setting, as it trivially holds due to the assumption $\bx^* \in {\rm Range}(G)$.  Hence, we only require the S-REC to ensure the following simplified form of \eqref{eq:upper_bds}:
    \begin{equation}\label{eq:upper_bds3}
        \|G(\hat{\bz})-\bx^*\|_2 \le 3\|\bmeta\|_2 + 2\delta.
    \end{equation}
    Since this holds with high probability (and uniformly in $\bx^*$ and $\bmeta$) for random $\bA$, it also holds for the best possible fixed $\bA$.  In addition, by standard $\chi^2$ concentration \cite[Sec.~2.1]{wainwright2019high}, we may assume that this fixed $\bA$ satisfies $\|\bA\|_\rmF^2 = n(1+o(1))$.  Then, squaring both sides of \eqref{eq:upper_bds3} and averaging over the Gaussian noise $\bmeta\sim \calN(0,\frac{\alpha}{m} \bI_m)$, we obtain a guarantee of the form \eqref{eq:cor_guarantee}.

    The preceding argument gives $\|\bA\|_\rmF^2 = n(1+o(1))$, but we can also normalize the measurement matrix to ensure $\|\bA\|_\rmF^2 = n$, and we see from the above definition of the S-REC that this only amounts to multiplying the parameters $(\gamma,\delta)$ by $1+o(1)$, which only amounts to changing the constant $C'$ in \eqref{eq:cor_guarantee}.  In addition, for rectangular domains, ($\bz \in B_{\infty}^k(r)$), the analogous claim us proved via Theorem \ref{thm:bora2} in an identical manner.

    %

\end{proof}


\section{Lower Bound for Bounded Lipschitz-Continuous Models} \label{sec:lower}

In this section, we construct a Lipschitz-continuous generative model that can generate bounded $k$-group-sparse vectors. Then, by making use of minimax statistical analysis for group-sparse recovery, we provide information-theoretic lower bounds that match the upper bounds in Corollary \ref{cor:bora}.

\subsection{Choice of Generative Model for the Rectangular Domain $B_{\infty}^k(r)$}
\label{sec:gen_model}

We would like to construct an $L$-Lipschitz function $G: B_{\infty}^k(r) \rightarrow \bbR^n$ such that recovering an arbitrary $\bx^*$ in its range with high probability and with $O(\alpha)$ squared error requires $m = \Omega\big(k\log \frac{Lr\sqrt{k}}{\sqrt{\alpha}}\big)$.  Recall that we consider $\by = \bA \bx^* + \bmeta$ with $A_{ij} \sim \calN(0,\frac{1}{m})$ and $\bmeta\sim \calN(0,\frac{\alpha}{m} \bI_m)$. 

\begin{figure}
    \begin{center}
        \includegraphics[width=0.45\textwidth]{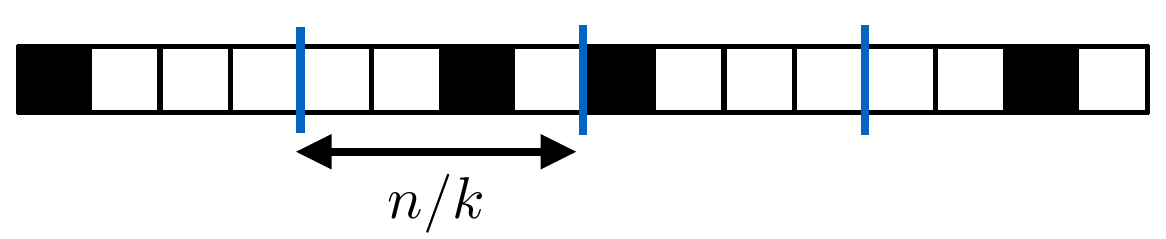}
    \end{center}
    \caption{Example of a $k$-group-sparse signal with length $n=16$ and sparsity $k=4$.  Each segment of size $\frac{n}{k} = 4$ has one non-zero entry, indicated by being filled instead of blank.} \label{fig:group} 
\end{figure}

Our approach is to construct such a generative model that is able to generate group-sparse signals, and then follow the steps of the minimax converse for sparse estimation \cite{Can13,Duc13}. More precisely, we say that a signal in $\bbR^n$ is {\em $k$-group-sparse} if, when divided into $k$ blocks of size $\frac{n}{k}$,\footnote{To simplify the notation, we assume that $n$ is an integer multiple of $k$.  For general values of $n$, the same analysis goes through by letting the final $n - k\lfloor\frac{n}{k}\rfloor$ entries of $\bx$ always equal zero.} each block contains at most one non-zero entry.\footnote{More general notions of group sparsity exist, but for compactness we simply refer to this specific notion as $k$-group-sparse.}  See Figure \ref{fig:group} for an illustration. We define
\begin{equation}
    \calS_k = \big\{ \bx \in \bbR^n: \bx \text{ is } k \text{-group-sparse} \big\},
\end{equation}
and the following constrained variants: 
\begin{gather}
     \calS_k(x_{\max}) = \big\{ \bx \in \calS_k \,:\, \|\bx\|_\infty \le x_{\max}  \big\},  \label{eq:S_group} \\
     \bar{\calS}_k(\xi) = \big\{ \bx \in \calS_k \,:\, \|\bx\|_0 = k,  \|\bx\|_{\infty} = \xi  \big\}.  \label{eq:S_group_eq}
\end{gather}
The vectors in $\bar{\calS}_k(\xi)$ have exactly $k$ non-zero entries all having magnitude $\xi$.  These vectors alone will suffice for establishing our lower bound (with a suitable choice of $\xi$), but we construct a generative model capable of producing all signals in $\calS_k(x_{\max})$; this is done as follows:
\begin{itemize}
    \item The output vector $\bx \in \bbR^n$ is divided into $k$ blocks of length $\frac{n}{k}$, denoted by $\bx^{(1)},\dotsc,\bx^{(k)} \in \bbR^{\frac{n}{k}}$.
    \item A given block $\bx^{(i)}$ is only a function of the corresponding input $z_i$, for $i=1,\dotsc,k$.
    \item The mapping from $z_i$ to $\bx^{(i)}$ is as shown in Figure \ref{fig:toy_gen}.  The interval $[-r,r]$ is divided into $\frac{n}{k}$ intervals of length $\frac{2rk}{n}$, and the $j$-th entry of $\bx^{(i)}$ can only be non-zero if $z_i$ takes a value in the $j$-th interval.  Within that interval, the mapping takes a ``double-triangular'' shape -- the endpoints and midpoint are mapped to zero, the points $\frac{1}{4}$ and $\frac{3}{4}$ of the way into the interval are mapped to $x_{\max}$ and $-x_{\max}$ respectively, and the remaining points follow a linear interpolation.  As a result, all values in the range $[-x_{\max},x_{\max}]$ can be produced.
\end{itemize}
\begin{figure}
    \begin{center}
        \includegraphics[width=1.0\columnwidth]{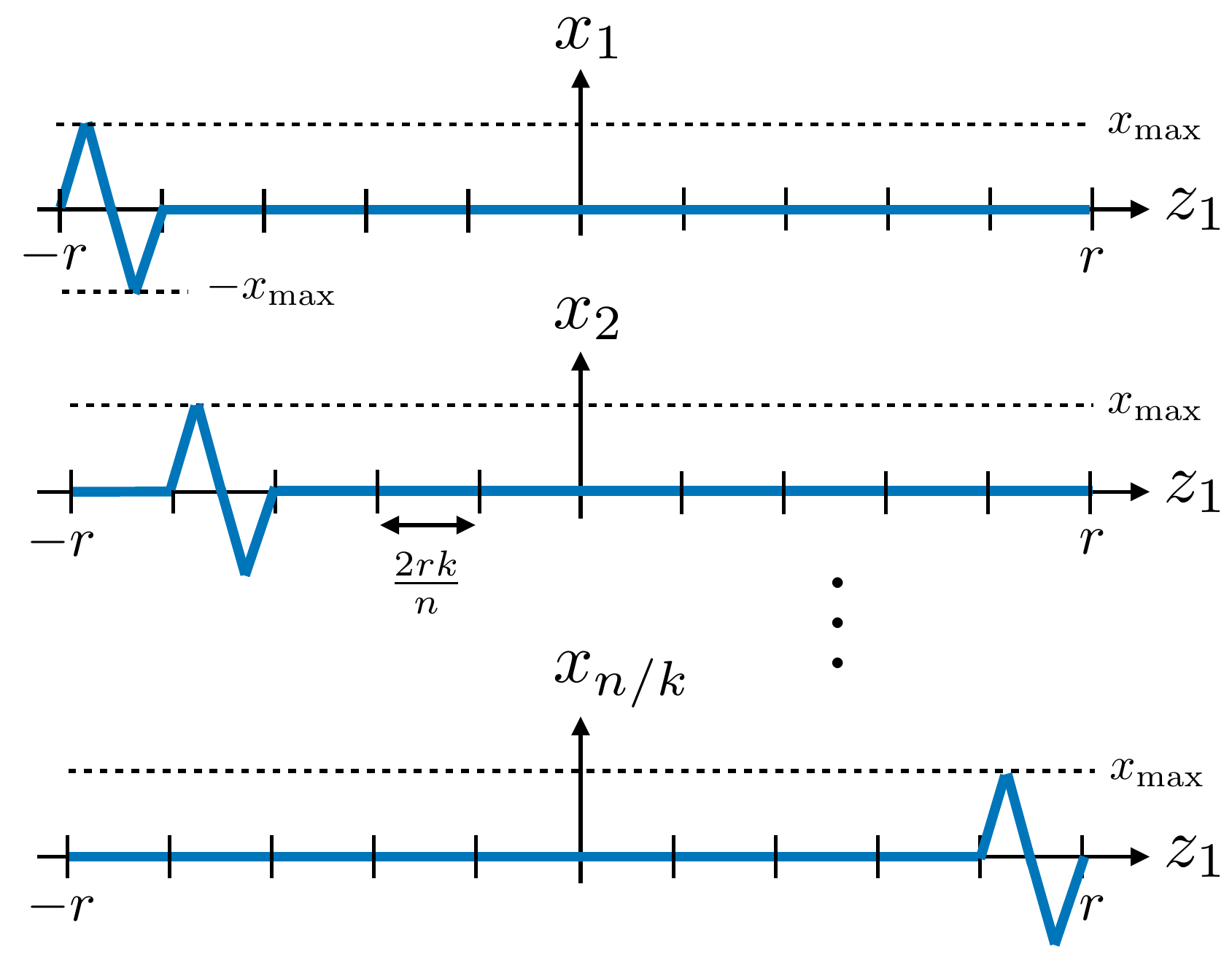}
    \end{center}
    \caption{Generative model that produces sparse signals.  This figure shows the mapping from $z_1 \to (x_1,\dotsc,x_{\frac{n}{k}})$, and the same relation holds for $z_2 \to (x_{\frac{n}{k}+1},\dotsc,x_{\frac{2n}{k}})$ and so on up to $z_k \to (x_{n-k+1},\dotsc,x_{n})$.} \label{fig:toy_gen} 
\end{figure}

While this generative model is considerably simpler than those used to generate complex synthetic data (e.g., natural images), it suffices for our purposes because it satisfies the assumptions imposed in \cite{Bor17}.  Our main goal is to show that the scaling laws of the upper bounds in \cite{Bor17} cannot be improved without further assumptions.

The simplicity of the preceding generative model permits a direct calculation of the Lipschitz constant, stated as follows.

\begin{lemma} \label{lem:Lipschitz}
    {\em (Calculation of Lipschitz Constant)}
    The generative model $G \,:\, B_{\infty}^k(r) \to \calS_k(x_{\max})$ described above, with parameters $n$, $k$, $r$, and $x_{\max}$, has a Lipschitz constant of
    \begin{equation}
        L = \frac{2nx_{\max}}{kr}.
    \end{equation}
\end{lemma}
\begin{proof}
    Recall that $\bx^{(i)}$ is the length-$\frac{n}{k}$ block corresponding to $z_i$, and for concreteness consider $i=1$. For two distinct $\bx^{(1)}, \bar{\bx}^{(1)} \in \bbR^{\frac{n}{k}}$, it is easy to see that the ratio $\frac{\|\bx^{(1)}- \bar{\bx}^{(1)}\|_2}{|z_1 -\bar{z}_1|}$ is maximized when $z_1$ and $\bar{z}_1$ are in the same small interval with length $\frac{2kr}{n}$. In particular, the Lipschitz constant for the sub-block is the absolute value of the slope of a line segment in that interval, denoted by $L_0 = \frac{2nx_{\max}}{kr}$. Then, combining the $k$ sub-blocks, we have
    \begin{align}
        \|\bx - \bar{\bx}\|_2^2 
        &= \sum_{i=1}^k \|\bx^{(i)} - \bar{\bx}^{(i)}\|_2^2 \\
        &\le L_0^2 \sum_{i=1}^k (z_1 - \bar{z}_1)^2 \\
        &=  L_0^2  \| \bz - \bar{\bz} \|_2^2,
    \end{align}
    so the overall Lipschitz constant is also $L = L_0 = \frac{2nx_{\max}}{kr}$. 
\end{proof}


\subsection{Minimax Lower Bound for Group-Sparse Recovery}

Consider the problem of estimating a $k$-group-sparse signal $\bx^* \in \bar{\calS}_k(\xi)$ (see \eqref{eq:S_group_eq}) from linear measurements $\by = \bA\bx^* + \bmeta$, where $\bmeta\sim \calN(0,\sigma^2 \bI_m)$ (we will later substitute $\sigma^2 = \frac{\alpha}{m}$).  
Specifically, given $\bA$ and $\by$, an estimate $\hat{\bx}$ is formed.  We are interested in establishing a lower bound on the minimax risk $\inf_{\hat{\bx}} \sup_{\bx^* \in \bar{\calS}_k(\xi)} \bbE_{\bx^*} \left[\|\bx^*-\hat{\bx}\|_2^2\right]$, where $\bbE_{\bx^*}$ denotes expectation when the underlying vector is $\bx^*$, and $\hat{\bx}$ is understood to be a function of $(\bA,\by)$, with $\by$ depending on $\bx^*$.


The following lemma states a minimax lower bound for $k$-group-sparse recovery under a suitable choice of $\xi$.  This result is proved using similar ideas to the case of regular $k$-sparse recovery \cite{Can13,Duc13}, with suitable modifications.

\begin{lemma}\label{lem:group}
    {\em (Lower Bound for Group-Sparse Recovery)}
    Consider the problem of $k$-group-sparse recovery with parameters $n$, $m$, $k$, and $\sigma^2$, with a given measurement matrix $\bA \in \bbR^{m \times n}$.  If $n \ge C_0 k$ for an absolute constant $C_0$, and if $\xi = \sqrt{\frac{n \sigma^2 \log\frac{n}{k}}{4\|\bA\|_\rmF^2}}$, then we have 
    \begin{equation}
        \inf_{\hat{\bx}} \sup_{\bx^* \in \bar{\calS}_k(\xi)} \bbE_{\bx^*} \left[\|\bx^*-\hat{\bx}\|_2^2\right] \ge \frac{n \sigma^2 k \log \frac{n}{k}}{64\|\bA\|_\rmF^2}. \label{eq:minimax_sigma}
    \end{equation}
    In particular, to achieve $\inf_{\hat{\bx}} \sup_{\bx^* \in \bar{\calS}_k(\xi)} \bbE_{\bx^*} \left[\|\bx^*-\hat{\bx}\|_2^2\right] \le C_1 m\sigma^2$ for a positive constant $C_1 > 0$, we require 
    \begin{equation}\label{eq:sample_complexity1}
        m \ge \frac{n k \log \frac{n}{k}}{ 64 C_1 \|\bA\|_\rmF^2}.
    \end{equation}
\end{lemma}
\begin{proof}
    See Appendix \ref{sec:pf_sparse}.
\end{proof}

Of course, \eqref{eq:minimax_sigma} trivially remains true when the supremum is taken over any set containing $\bar{\calS}_k(\xi)$, in particular including $\calS_k(x_{\max})$ for any $x_{\max} \ge \xi$.

\subsection{Statement of Main Result}

Combining the preceding auxiliary results, we deduce the following information-theoretic lower bound for compressive sensing with generative models having a rectangular domain.

\begin{theorem} \label{thm:main}
    {\em (Lower Bound for Rectangular Domains)}
    Consider the problem of compressive sensing with $L$-Lipschitz generative models with input domain $B_{\infty}^k(r)$, and i.i.d.~$\calN\big( 0,\frac{\alpha}{m} \big)$ noise.  Let $C_1 > 0$ and $C_{\bA} > 0$ be fixed constants, and assume that $L \ge \Omega\big( \frac{1}{r} \sqrt{\frac{\alpha}{k}} \big)$ with a sufficiently large implied constant.  Then there exists an $L$-Lipschitz generative model $G \,:\, B_{\infty}^k(r) \to \bbR^n$ (and associated output dimension $n$) such that, for any $\bA \in \bbR^{m \times n}$ satisfying $\|\bA\|_{\rmF}^2 = C_{\bA} n$, any algorithm that produces some $\hat{\bx}$ satisfying
    \begin{equation}
        \sup_{ \bx^* \in {\rm Range}(G) } \bbE\big[ \|\hat{\bx} - \bx^*\|_2^2 \big] \le C_1 \alpha \label{eq:final_guarantee}
    \end{equation}
    must also have $m = \Omega\big( k \log\frac{Lr \sqrt{k}}{\sqrt{\alpha}} \big)$.
\end{theorem}
\begin{proof}
    We are free to choose the output dimension $n$ to our liking for the purpose of proving the theorem, and accordingly, we set
    \begin{equation}
        n = \frac{C' L r k \sqrt{k}}{\sqrt{\alpha}} \label{eq:choice_n}
    \end{equation}
    for some constant $C'$ to be chosen later.  As a result, we have
    \begin{equation}
        k \log \frac{n}{k} = k \log\frac{C' Lr \sqrt{k}}{\sqrt{\alpha}} = \Theta\bigg( k \log\frac{Lr \sqrt{k}}{\sqrt{\alpha}}  \bigg)
    \end{equation}
    since we assumed that $L \ge \Omega\big( \frac{1}{r} \sqrt{\frac{\alpha}{k}} \big)$ with a sufficiently large implied constant.  Hence, it suffices to show that $m = \Omega\big( k \log \frac{n}{k} \big)$ is necessary for achieving \eqref{eq:final_guarantee}.

    To do this, we make use of Lemma \ref{lem:group} on $k$-group-sparse recovery, and the fact that our choice of generative model is able to produce such signals.  Since we assumed that $\|\bA\|_{\rmF}^2 = C_{\bA} n$, the contrapositive form of Lemma \ref{lem:group} states that under the assumptions therein, it is not possible to achieve \eqref{eq:final_guarantee} when
    \begin{equation}\label{eq:sample_complexity2}
        m < \frac{n k \log \frac{n}{k}}{ 64 C_1 C_{\bA}}.
    \end{equation}
    While this has the desired $k\log\frac{n}{k}$ behavior, the result only holds true under the conditions $n \ge C_0 k$ and $x_{\max} \ge \xi = \sqrt{\frac{\alpha \log\frac{n}{k}}{4 C_{\bA} m}}$ from Lemma \ref{lem:group} (after setting $\sigma^2 = \frac{\alpha}{m}$ and $\|\bA\|_\rmF^2 = C_{\bA} n$).  We proceed by checking that the assumptions of Theorem \ref{thm:main} imply that both of these conditions are true.

    The condition $n \ge C_0 k$ follows directly from \eqref{eq:choice_n} and the assumption that $L \ge \Omega\big( \frac{1}{r} \sqrt{\frac{\alpha}{k}} \big)$ with a sufficiently large implied constant.  For the condition on $x_{\max}$, we equate the condition $L = \frac{ n\sqrt{\alpha} }{ C' r k\sqrt{k} }$ from \eqref{eq:choice_n} with the finding $L = \frac{2nx_{\max}}{kr}$ from Lemma \ref{lem:Lipschitz}; canceling the $\frac{n}{kr}$ terms and re-arranging gives
    \begin{equation}
        x_{\max} =  \frac{\sqrt{\alpha}}{2 C' \sqrt{k}}. \label{eq:choice_x_max}
    \end{equation}
    As a result, we have the required condition $x_{\max} \ge \sqrt{\frac{\alpha \log\frac{n}{k}}{4  C_{\bA} m}}$ as long as
    \begin{equation}
        m \ge \frac{ (C')^2 }{ C_{\bA} }  k \log \frac{n}{k}. \label{eq:m_cond}
    \end{equation}
    Hence, we have shown that it is impossible to achieve \eqref{eq:final_guarantee} in the case that both \eqref{eq:sample_complexity2} and \eqref{eq:m_cond} hold.  To make these two conditions consistent, we set $C' = \frac{1}{\sqrt{128 C_1}}$, meaning that \eqref{eq:m_cond} reduces to $m \ge \frac{ 1 }{ 128 C' C_{\bA} }  k \log \frac{n}{k}$.

    The preceding findings show, in particular, that if $m$ is the largest integer satisfying \eqref{eq:sample_complexity2} (henceforth denoted by $m^*$), then it is impossible to achieve \eqref{eq:final_guarantee}.  To show that the same is true for all smaller values of $m$, we use the simple fact that additional measurements can only ever help.  More formally, suppose that $\tilde{\bA}$ is an $m_0 \times n$ measurement matrix achieving \eqref{eq:final_guarantee} for some $m_0 < m^*$.  Consider adding $m^* - m_0$ rows of zeros to $\tilde{\bA}$ to produce $\bA$, so that $\|\bA\|_{\rmF} = \|\tilde{\bA}\|_{\rmF}$.  If one ignores the final $m^* - m_0$ entries of $\by$, then the problem of recovery from $m^*$ measurements is reduced to that from $m_0$ measurements.  In fact, in the latter case, the noise variance is also reduced to $\frac{\alpha}{m^*} < \frac{\alpha}{m_0}$, but to precisely recover the desired setting corresponding to $m_0$ measurements, the recovery algorithm can artificially add $\calN\big(0, \frac{\alpha}{m_0} - \frac{\alpha}{m^*} \big)$ noise to each entry.
\end{proof}

Theorem \ref{thm:main} shows that the scaling laws of Corollary \ref{cor:bora} cannot be improved, regardless of the measurement matrix and recovery algorithm.  The result holds under the assumption that $L \ge \Omega\big( \frac{1}{r} \sqrt{\frac{\alpha}{k}} \big)$ with a sufficiently large implied constant, which is a very mild assumption since for fixed $r$ and $\alpha$, the right-hand side tends to zero as $k$ grows large (whereas typical Lipschitz constants are at least equal to one, if not much higher).\footnote{\label{foot:L_cond}In fact, if we were to have $L = o\big( \frac{1}{r} \sqrt{\frac{\alpha}{k}} \big)$, then the $O\big( k \log\frac{Lr \sqrt{k}}{\sqrt{\alpha}} \big)$ scaling of Corollary \ref{cor:bora} would seemingly not make sense.  The explanation is that in this regime, outputting {\em any} $\hat{\bz} \in B_{\infty}^k(r)$ suffices for the recovery guarantee, and no measurements are needed at all.}

\begin{remark} \label{rem:forall}
    For convenience, we stated Theorem \ref{thm:main} in terms of the supremum (worst-case error) over $\bx^* \in {\rm Range}(G)$, but in fact Lemma \ref{lem:group} is based on lower bounding the worst-case minimax risk by the average over a suitably-chosen distribution on $\bx^*$.  As a result, Theorem \ref{thm:main} remains true when \eqref{eq:final_guarantee} is replaced by the guarantee
    \begin{equation}
        \bbE\big[ \|\hat{\bx} - \bx^*\|_2^2 \big] \le C_1 \alpha \label{eq:final_guarantee_foreach}
    \end{equation}
    with the average being over both $\bx^*$ (with the ``hard'' prior distribution not depending on $\bA$) and the noise.  Hence, while Theorem \ref{thm:main} is stated as a hardness result on a ``for-all'' style (worst-case) guarantee with respect to $\bx^*$, the preceding variant proves the hardness of a ``for-each'' style (average-case) result \cite{Can11}.  In the general setup of Theorems \ref{thm:bora1} and \ref{thm:bora2} where $\bx^*$ may have representation error, it is in fact essential to consider the less stringent ``for-each'' style guarantee, whereas Corollary \ref{cor:bora} reveals that the ``for-all'' style guarantee is indeed possible when there is no representation error.
\end{remark}


\subsection{Extension to the Spherical Domain $B_{2}^k(r)$}

The above analysis focuses on the rectangular domain $B_{\infty}^k(r)$.  At first glance, it may appear non-trivial to use the same ideas to obtain corresponding lower bounds for the spherical domain $B_{2}^k(r)$.  However, in the following, we show that by simply considering the largest possible $\ell_{\infty}$-ball inside the $\ell_{2}$-ball, we can obtain a matching lower bound to Corollary \ref{cor:bora} even for spherical domains.  The fact that this crude approach gives a tight result may appear to be surprising, and is discussed further below.

Let $G_{{\rm rect}, r}$ denote the above-formed generative model for rectangular domains with radius $r$, and note that $\calB_\infty^k\big(\frac{r}{\sqrt k}\big) \subseteq \calB_2^k(r) \subseteq \calB_\infty^k(r)$. To handle the spherical domain $B_{2}^k(r)$, we construct the generative model $G(\cdot)$ as follows: 
\begin{itemize}
    \item For any $\bz \in \calB_\infty^k\big(\frac{r}{\sqrt k}\big)$, we simply let $G(\bz) = G_{{\rm rect}, \frac{r}{\sqrt k}}(\bz)$.  It is only these input values that will be used to establish the lower bound, as these values alone suffice for generating all signals in $\calS_k(x_{\max})$.  However, we still need to set the other values to ensure that Lipschitz continuity is maintained.
    \item To handle the other values of $\bz$, we extend the functions in Figure \ref{fig:toy_gen} (with $\frac{r}{\sqrt k}$ in place of $r$) to take values on the whole real line:  For all values outside the indicated interval, each function value simply remains zero.
    \item The preceding dot point leads to a Lipschitz-continuous function defined on all of $\bbR^k$, and we simply take $G(\cdot)$ to be that function restricted to $\calB_\infty^k(r)$.
%
\end{itemize}
By the first dot point above, we can directly apply Theorem \ref{thm:main} with $\frac{r}{\sqrt k}$ in place of $r$, yielding the following.

\begin{theorem} \label{thm:main2}
    {\em (Lower Bound for Spherical Domains)}
    Consider the problem of compressive sensing with $L$-Lipschitz generative models, with input domain $B_{2}^k(r)$ and i.i.d.~$\calN\big( 0,\frac{\alpha}{m} \big)$ noise.  Let $C_1 > 0$ and $C_{\bA} > 0$ be fixed constants, and assume that $L \ge \Omega\big( \frac{\sqrt \alpha}{r} \big)$ with a sufficiently large implied constant.  Then there exists an $L$-Lipschitz generative model $G\,:\,B_2^{k}(r) \to \bbR^n$ (and associated output dimension $n$) such that, for any $\bA \in \bbR^{m \times n}$ satisfying $\|\bA\|_{\rmF}^2 = C_{\bA} n$, any algorithm that produces some $\hat{\bx}$ satisfying
    \begin{equation}
        \sup_{ \bx^* \in {\rm Range}(G) } \bbE\big[ \|\hat{\bx} - \bx^*\|_2^2 \big] \le C_1 \alpha \label{eq:final_guarantee2}
    \end{equation}
    must also have $m = \Omega\big( k \log\frac{Lr}{\sqrt{\alpha}} \big)$.
\end{theorem}

This result establishes the tightness of Corollary \ref{cor:bora} up to constant factors for spherical domains.  The assumption $L \ge \Omega\big( \frac{\sqrt \alpha}{r} \big)$ is different from that of Theorem \ref{thm:main}, but is similarly mild (see the discussion in Footnote \ref{foot:L_cond} on Page \pageref{foot:L_cond}).

The above reduction may seem overly crude, because as $k$ grows large the volume of $B_{\infty}^k\big( \frac{r}{\sqrt k} \big)$ is a vanishingly small fraction of the volume of $B_{2}^k(r)$.  However, as discussed following Theorem \ref{thm:bora1}, the key geometric quantity in the proof of the upper bound is in fact the covering number (see also Appendix \ref{sec:covering}), and both $B_{2}^k(r)$ and $\calB_\infty^k\big(\frac{r}{\sqrt k}\big)$ yield the same scaling laws for the logarithm of the covering number (with a sufficiently small distance parameter). As a result, it is reasonable to expect that these two domains also require the same scaling laws on the number of measurements.


\section{Generative Models Based on ReLU Networks} \label{sec:relu}

In this section, as opposed to considering general Lipschitz-continuous generative models, we provide a more detailed treatment of neural networks with ReLU activations (see Appendix \ref{sec:relu_defs} for a brief overview of the relevant definitions).
We are particularly interested in comparing against the following result from \cite{Bor17}; this result holds even when the domain is unbounded ($\calD = \bbR^k$), so we do not need to distinguish between rectangular and spherical domains.

\begin{theorem} {\em (Upper Bound for ReLU Networks \cite[Thm.~1.1]{Bor17})} \label{thm:bora3}
    Let $G \,:\, \bbR^k \to \bbR^n$ be a generative model given by a $d$-layer neural network with ReLU activations\footnote{As discussed in \cite{Bor17}, the same result holds for any piecewise linear activation function with two components (e.g., leaky ReLU).} and at most $w$ nodes per layer (with $w \ge 2$), and let $\bA \in \bbR^{m \times n}$ be a random measurement matrix whose entries are i.i.d.~with distribution $A_{ij} \sim \calN(0,\frac{1}{m})$. Given the observed vector $\by = \bA\bx^* + \bmeta$ for some $\bx^*$, let $\hat{\bz}$ minimize $\|\by - \bA G(\bz)\|_2$ to within additive error $\epsilon$ of the optimum over $\bz \in \bbR^k$. Then, with a number of measurements satisfying $m \ge C k d\log w$ for a universal constant $C$, the following holds with probability $1-e^{-\Omega(m)}$:
    \begin{equation} 
        \|G(\hat{\bz})-\bx^*\|_2 \le 6 \min_{\bz^* \in \bbR^k} \|G(\bz^*)-\bx^*\|_2 + 3\|\bmeta\|_2 + 2\epsilon. \label{eq:relu_ach}
    \end{equation}
\end{theorem}

It is interesting to note that this result makes no assumptions about the neural network weights (nor domain size), but rather, only the input size, width, and depth.  In addition, we have the following counterpart to Corollary \ref{cor:bora}, whose proof is essentially the same as Corollary \ref{cor:bora} and is therefore omitted.

\begin{corollary} \label{cor:bora3}
    {\em (Simplified Upper Bound for ReLU Networks)}
    Consider the setup of Theorem \ref{thm:bora3} with fixed $\bA$ in place of Gaussian measurements, and with $\bx^* \in {\rm Range}(G)$, no optimization error ($\epsilon = 0$), and i.i.d.~Gaussian noise  $\bmeta\sim \calN(0,\frac{\alpha}{m} \bI_m)$ with $\alpha > 0$.  Then, when $m \ge C k d\log w$ for a universal constant $C$, there exists some $\bA$ with $ \|\bA\|_\rmF^2 = n$ such that
    \begin{equation}
        \sup_{ \bx^* \in {\rm Range}(G) } \bbE\big[ \|G(\hat{\bz}) - \bx^*\|_2^2 \big] \le C' \alpha \label{eq:relu_cor}
    \end{equation}
    for a universal constant $C'$.
\end{corollary}

Before establishing corresponding lower bounds, it is useful to first discuss how the generative model from Figure \ref{fig:toy_gen} can be constructed using ReLU networks; this is done in Section \ref{sec:constructing}.  In Section \ref{sec:width_depth}, we build on these ideas to form different (but related) generative models that properly reveal the dependence of the number of measurements on the network depth and width.

\subsection{Constructing the Generative Model Used in Theorem \ref{thm:main}} \label{sec:constructing}

In the case of a rectangular domain, the triangular shapes of the mappings in Figure \ref{fig:toy_gen} are such that the generative model $G(z)$ can {\em directly} be implemented as a ReLU network with a single hidden layer (i.e., two layers in total).  Indeed, this would remain true if the mappings between $z_i$ and $x_j$ (with $x_j$ being a single entry of $\bx^{(i)}$) in Figure \ref{fig:toy_gen} were replaced by {\em any} piecewise linear function \cite[Thm.~2.2]{Aro18}.

A limitation of this interpretation as a two-layer ReLU network is that for increasing values of $L$, the corresponding network has increasingly large weights.  In particular, for fixed values of $r$ and $\alpha$, a re-arrangement of \eqref{eq:choice_n} gives $L = \Theta\big( \frac{n}{k\sqrt{k}} \big)$, which amounts to large weights in the case that $n \gg k\sqrt{k}$.

 In the following, we argue that the construction of Figure \ref{fig:toy_gen} can be implemented using a {\em deep} ReLU network with {\em bounded} weights.  To see this, we use similar ideas to those used to generate rapidly-varying (e.g., ``sawtooth'') functions using ReLU networks \cite{Tel15,Tel16} (see also Section \ref{sec:width_depth}). 

Consider the functions $f(z)$ and $g(z)$ shown in Figure \ref{fig:deep}. If we compose $g(\cdot)$ with itself $D$ times, then we obtain a function equaling $-r$ for $z \in [-r,-r 2^{-D}]$, equaling $r$ for $z \in [r2^D,r]$, and linearly interpolating in between.  By further composing this function with $f(\cdot)$, we obtain a function of the form shown in Figure \ref{fig:deep} (Right), which matches those in Figure \ref{fig:toy_gen}.  By incorporating suitable offsets into this procedure, one can obtain the same ``double triangular'' shape shifted along the horizontal axis, and hence recover all of the mappings shown in Figure \ref{fig:toy_gen}.

Since both $f$ and $g$ are piecewise linear with slope at most $2$, both of these functions can be implemented with a single hidden layer with $O(1)$-bounded weights and $O(r)$-bounded offsets.  To bring the $2^{-D}$-width ``double triangular'' region down to the desired width $\frac{2rk}{n}$ in Figure \ref{fig:toy_gen}, we take $D = O(\log \frac{n}{kr})$ compositions of $g$ (each of which adds another two layers to the network).\footnote{The case that $\frac{n}{2rk}$ is not a power of two can be handled by slightly modifying the function $g(\cdot)$ in Figure \ref{fig:deep}, i.e., moving the changepoints currently occurring at $z = -r/2$ and $z = r/2$.}  Then, recall from Figure \ref{fig:toy_gen} that the final generative model consists of $n$ one-dimensional functions; we let the network incorporate these in parallel.  Combining these findings, we have the following.

\begin{theorem} \label{thm:bounded}
    {\em (Construction of Generative Model Using a ReLU Network)}
    Consider the setup of Theorem \ref{thm:main}, and suppose that $x_{\max} = O(1)$. Then, the generative model therein can be realized using a ReLU network with depth $O(\log \frac{n}{kr})$, width $O(n)$, weights bounded by $O(1)$, and offsets bounded by $O(r)$.
\end{theorem}

Note that the assumption $x_{\max} = O(1)$ is mild in view of \eqref{eq:choice_x_max}, and even if one wishes to handle more general values, it is not difficult to generalize the above arguments accordingly.

\begin{figure*}
    \begin{center}
        \includegraphics[width=0.7\textwidth]{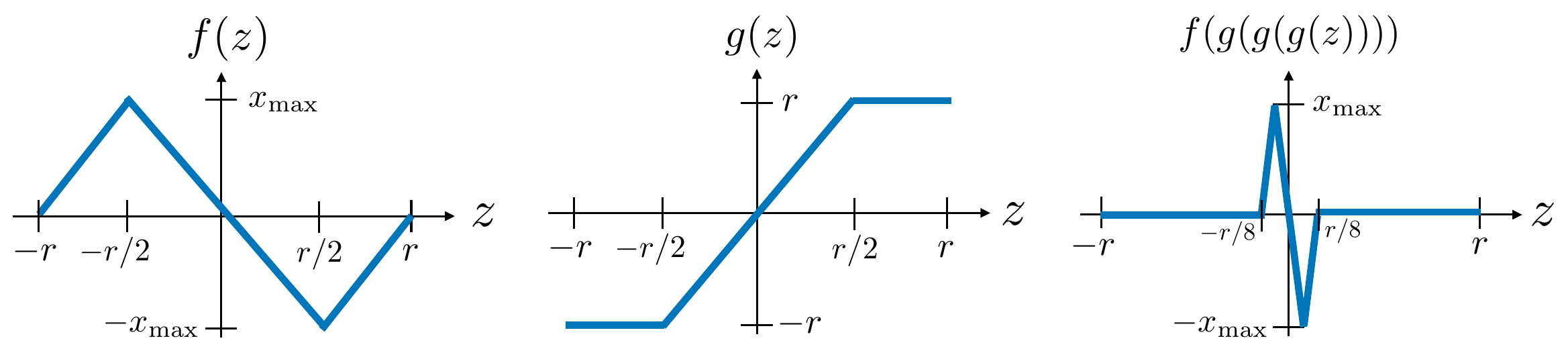}
    \end{center}
    \caption{Illustration of how the functions in Figure \ref{fig:toy_gen} can be generated using a deep ReLU network with bounded weights.  The right-most function is centered at zero, but by including additional offset terms, this can be shifted to different locations in $[-r,r]$.} \label{fig:deep} 
\end{figure*}

\subsection{Understanding the Dependence on Width and Depth} \label{sec:width_depth}

Thus far, we have considered forming a generative model $G \,:\, \bbR^k \to \bbR^n$ capable of producing $k$-group-sparse signals, which leads to a lower bound of $m = \Omega(k \log n)$.  While this precise approach does not appear to be suited to properly understanding the dependence on depth and width in Theorem \ref{thm:bora3}, we now show that a variant indeed suffices: We form a wide and/or deep ReLU network $G \,:\, \bbR^k \to \bbR^n$ capable of producing all $(kk_0)$-group-sparse signals for some $k_0$ that may be much larger than one. 

It is instructive to first consider the case $k = 1$ and $k_0 > 1$, and to construct a {\em non-continuous} generative model that will later be slightly modified to be continuous.  For later convenience, we momentarily denote the output length by $n_0$.  We consider the interval $[0,1]$, which we view as being split into $\frac{1}{ 2^{k_0} (\frac{n_0}{k_0})^{k_0} }$ small intervals of equal length; note that $2^{k_0} (\frac{n_0}{k_0})^{k_0}$ is the number of possible {\em signed} sparsity patterns for group-sparse signals of length $n_0$ with exactly $k_0$ non-zero entries.  The idea is to let each value of $z$ corresponding to the mid-point of a given length-$\frac{1}{ 2^{k_0} (\frac{n_0}{k_0})^{k_0} }$ interval in $[0,1]$ produce a signal with a different signed sparsity pattern.  

In more detail, we consider the following (see Figure \ref{fig:toy_gen2} for an illustration):
\begin{itemize}
    \item (Coarsest scale) The interval $[0,1]$ is split into $\frac{2n_0}{k_0}$ intervals of length $\frac{k_0}{2n_0}$.  Then:
    \begin{itemize}
        \item If $z$ lies in the first interval, we have $x_1 = \xi$, and if $z$ lies in the second interval, we have $x_1 = - \xi$ (in all other cases, $x_1 = 0$);
        \item If $z$ lies in the third interval, we have $x_2 = \xi$, and if $z$ lies in the forth interval, we have $x_2 = - \xi$ (in all other cases, $x_2 = 0$);
        \item This continues similarly for $x_3,\dotsc,x_{\frac{n_0}{k_0}}$.
    \end{itemize}
    \item (Second coarsest scale) Each interval at the coarsest scale is split into  $\frac{2n_0}{k_0}$ equal sub-intervals of length $\big(\frac{k_0}{2n_0}\big)^2$.  Then, within each of the coarsest intervals:
    \begin{itemize}
        \item If $z$ lies in the first sub-interval, we have $x_{\frac{n_0}{k_0} + 1} = \xi$, and if $z$ lies in the second sub-interval, we have $x_{\frac{n_0}{k_0} + 1} = - \xi$ (in all other cases, $x_{\frac{n_0}{k_0} + 1} = 0$);
        \item If $z$ lies in the third sub-interval, we have $x_{\frac{n_0}{k_0} + 2} = \xi$, and if $z$ lies in the forth sub-interval, we have $x_{\frac{n_0}{k_0} + 2}  = - \xi$ (in all other cases, $x_{\frac{n_0}{k_0} + 2}  = 0$);
        \item This continues similarly for $x_{\frac{n_0}{k_0} + 3} ,\dotsc,x_{\frac{2n_0}{k_0}}$.
    \end{itemize}
    \item We continue recursively until we are at the finest scale with sub-intervals of length $\big(\frac{k_0}{2n_0}\big)^{k_0}$ that dictate the values of $x_{\frac{(k_0-1)n_0}{k_0} + 1} ,\dotsc,x_{n_0}$.
\end{itemize}

\begin{figure*}
    \begin{center}
        \includegraphics[width=0.87\textwidth]{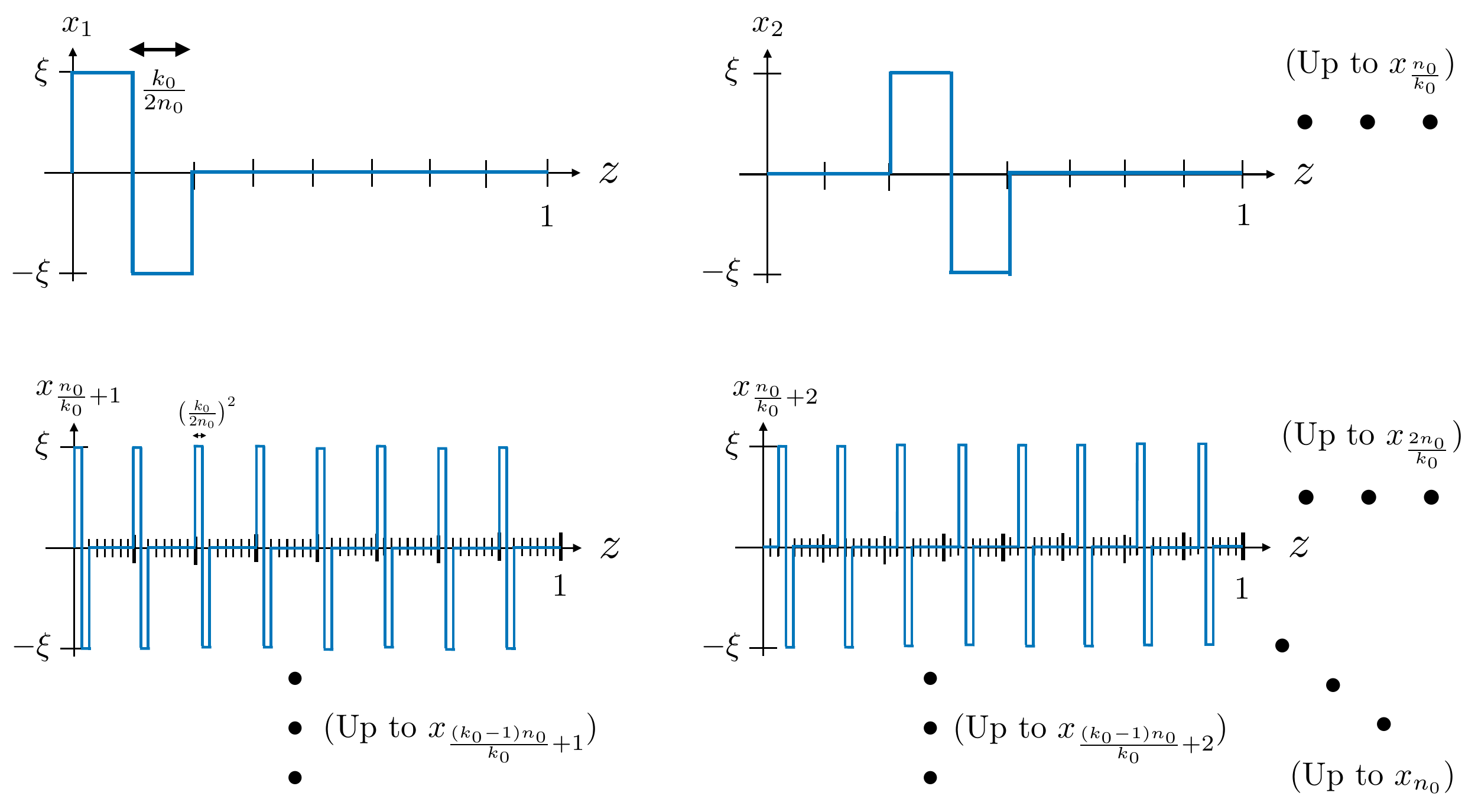}
    \end{center}
    \caption{Recursively defined generative model mapping a single input $z$ to $\bbR^{n_0}$.  Note that this figure depicts an idealized version in which the mapping is discontinuous; the final generative model used replaces each vertical line by a sharp (but finite-gradient) transition, creating pulses that are trapezoidal rather than rectangular.} \label{fig:toy_gen2} 
\end{figure*}

While the discontinuous points in Figure \ref{fig:toy_gen2} are problematic when it comes to implementation with a ReLU network, we can overcome this by simply replacing them by straight-line transitions having a finite slope (i.e., the rectangular shapes become trapezoidal), while being sufficiently sharp so that all the input values at the midpoints of the length-$\frac{1}{ 2^{k_0} (\frac{n_0}{k_0})^{k_0} }$ intervals produce the same outputs as the idealized function described above.
Then, ReLU-based implementation is mathematically possible, since the mappings are piecewise linear \cite[Thm.~2.2]{Aro18}.

The above construction generates all $k_0$-group-sparse signals in $\bbR^{n_0}$ with non-zero entries equaling $\pm \xi$.  To see this, one can consider ``entering'' the appropriate coarsest region according to the desired location and sign in the first block (of length $\frac{n_0}{k_0}$) of the $k_0$-sparse signal, then recursively entering the appropriate second-coarsest region based on the second block, and so on.  The final $z$ value reached will correspond to the desired $k_0$-group-sparse signal $\bx_0 \in \bbR^{n_0}$.

To generalize the above ideas to $k$-input generative models, we form $k$ such functions in parallel, thereby allowing the generation of $(kk_0)$-group-sparse signals in $\bbR^n$ (with $n = n_0 k$) having non-zero entries $\pm \xi$.  Then, we can use Lemma \ref{lem:group} and a suitable choice of $\xi$ to deduce the following.

\begin{theorem} \label{thm:wide_deep}
    {\em (Lower Bound for ReLU Networks)}
    Fix $C_1,C_{\bA} > 0$, and consider the problem of compressive sensing with generative models under i.i.d.~$\calN\big( 0,\frac{\alpha}{m} \big)$ noise, a measurement matrix $\bA \in \bbR^{m \times n}$ satisfying $\|\bA\|_{\rmF}^2 = C_{\bA} n$, and the above-described generative model $G \,:\, \bbR^k \to \bbR^n$ with parameters $k$, $k_0$, $n_0$, and $\xi$.  Then, if $n_0 \ge C_0 k_0$ for an absolute constant $C_0$, then there exists a constant $C_2 = \Theta(1)$ such that the choice $\xi = \sqrt{\frac{ C_2 \alpha }{ k }}$ yields the following:
    \begin{enumerate}
        \item Any algorithm that produces some $\hat{\bx}$ satisfying $\sup_{ \bx^* \in {\rm Range}(G) } \bbE\big[ \|\hat{\bx} - \bx^*\|_2^2 \big] \le C_1 \alpha$ must also have $m = \Omega\big( kk_0 \log\frac{n}{kk_0} \big)$ (or equivalently $m = \Omega\big( kk_0 \log\frac{n_0}{k_0} \big)$, since $n = n_0 k$).
        \item The generative function $G$ can be implemented as a ReLU network with any of the following combinations of the depth $d$ and width $w$:\footnote{More precisely, in the second and third cases here, we only require $G$ to match the description above for $\bz \in [0,1]^k$.  The generative model may take non-zero values outside this range, without affecting the first part of the theorem. This technicality arises because we will construct periodic functions of the kind in Figure \ref{fig:toy_gen2} for which the number of repetitions is a power of two, whereas $\frac{n_0}{k_0}$ may not be a power of two. \label{foot:domain}}
        \begin{enumerate}
            \item $d=2$ and $w = O( k (\frac{n_0}{k_0})^{k_0} )$;
            \item $d = O\big( k_0 \log \frac{n_0}{k_0} \big)$ and $w = O(n)$;
            \item $d$ exhibits arbitrary scaling subject to $d = \omega(1)$ and $d = o\big( k_0 \log \frac{n_0}{k_0} \big)$, and $w$ is chosen as a function of $d$ such that $\log w = O\big( \frac{k_0}{d} \log \frac{n_0}{k_0} \big) + \log n$.
        \end{enumerate}
        These three cases respectively give (a) $k d \log w = O\big( k k_0 \log \frac{n_0}{k_0} + k \log k \big)$; (b) $kd \log w = O\big( k k_0 \log \frac{n_0}{k_0} \cdot \log n \big)$; and (c) $kd \log w = O(k k_0 \log \frac{n_0}{k_0} + kd \log n)$.
    \end{enumerate}
\end{theorem}
\begin{proof}
    The first claim is proved similarly to the proof of Theorem \ref{thm:main}, so we only outline the differences.  In accordance with Lemma \ref{lem:group} (with $kk_0$ in place of $k$), let $m^*$ be the largest integer smaller than $\frac{kk_0 \log \frac{n}{kk_0}}{64 C_1 C_{\bA}}$.  Then, Lemma \ref{lem:group} states that if $m = m^*$ and $\xi = \sqrt{ \frac{\alpha \log \frac{n}{kk_0}}{4m^*} }$, it is not possible to achieve $\sup_{ \bx^* \in {\rm Range}(G) } \bbE\big[ \|G(\hat{\bz}) - \bx^*\|_2^2 \big] \le C_1 \alpha$.  Substituting this definition of $m^*$ into this choice of $\xi$ gives the claimed behavior $\xi = \sqrt{\frac{ C_2 \alpha }{ k }}$ with $C_2 = \Theta(1)$.  The first claim follows by using the argument at the end of the proof of Theorem \ref{thm:main} to argue that since the recovery goal cannot be attained when $m = m^*$, it also cannot be attained when $m < m^*$.

    For part (a) of the second claim, we observe that each mapping from $z$ to $x_i$ in Figure \ref{fig:toy_gen2} has a bounded number of ``pulses'', and at the $\ell$-th scale, the number of pulses is $2 \big( \frac{n_0}{k_0} \big)^{\ell-1}$.  Summing over $\ell = 1,\dotsc,k_0$ gives a total of at most $2 \big( \frac{n_0}{k_0} \big)^{k_0}$ pulses.  Recall also that these rectangles are replaced by trapezoidal shapes to make them implementable.  Hence, we can apply the well-known fact that any piecewise-linear function with $N$ pieces can be implemented using a ReLU network of width $O(N)$ with a single hidden layer \cite[Thm.~2.2]{Aro18}, and in our case we have $N = O\big( \big( \frac{n_0}{k_0} \big)^{k_0}  \big)$.  The desired claim follows by multiplying by $k$ in accordance with the fact that we implement the network of Figure \ref{fig:toy_gen2} in parallel $k$ times.

    Due to the periodic nature of the signals in Figure \ref{fig:toy_gen2}, part (b) also follows using well-established ideas \cite{Tel15,Tel16}.  We would like to produce trapezoidal pulses at regular intervals similarly to Figure \ref{fig:toy_gen2}.  To obtain the positive pulses, we can take a half-trapezoidal shape of the form in Figure \ref{fig:sawtooth} (Right) and compose it with a sawtooth function having some number $R$ of triangular regions as in Figure \ref{fig:sawtooth} (Middle), possibly using suitable offsets to shift the location.  The negative pulses can be produced similarly, and the two can be added together in the final layer.

    As exemplified in Figure \ref{fig:sawtooth} and proved in \cite{Tel15}, the $R$-piece sawtooth function itself can be implemented by a network with width $O(1)$ and depth $O( \log R )$ when $R$ is a power of two.  
    In our case, the maximal number of such repetitions is $R = \big( \frac{n_0}{k_0} \big)^{k_0-1}$ (at the finest scale), and if this is a power of two, the depth required is $d = O(\log R) = O\big( k_0 \log \frac{n_0}{k_0} \big)$.\footnote{At coarser scales, the required depth for the given mapping may be much smaller than the maximum depth corresponding to the finest scale.  To equalize all of the depths, we can simply let the earliest layers implement the desired mapping, and then repeatedly use the identity map (which can be implemented with a width of two since $\max\{0,z\} + \max\{0,-z\} = z$) in further layers. }  If $R$ is not a power of two, we can round $R$ up to the nearest power of two and extend the $z$ values slightly beyond the interval $[0,1]$ accordingly (see Footnote \ref{foot:domain}).  By the above-mentioned result of \cite{Tel15} on the sawtooth function, the width is a constant multiple of the number of outputs, i.e., $w = O(n)$.

    For part (c), the argument is similar to part (b), but exploits both depth and width to produce the $R$-piece sawtooth function.  We first exploit width similarly to part (a) above: At the finest scale, for a fixed constant $\beta > 0$, we consider producing a sawtooth function $h(\cdot)$ with $\Theta\big(\big( \frac{n_0}{k_0} \big)^{\frac{k_0}{\beta d}}\big)$ repetitions, using a depth-2 ReLU network with width $\Theta\big(\big( \frac{n_0}{k_0} \big)^{\frac{k_0}{\beta d}}\big)$.  Then, we exploit depth similarly to part (b) above: Composing $h(\cdot)$ with itself $\frac{d}{2}$ times, the number of repetitions gets raised to the power of $\Theta(d)$, and we get the desired number $R = \big( \frac{n_0}{k_0} \big)^{k_0-1}$ of repetitions when $\beta$ is suitably chosen.  Since this procedure is used to produce $n$ signals (including those at the coarser scales), the total width required is at most $w = O\big( n \big( \frac{n_0}{k_0} \big)^{\frac{k_0}{\beta d}}\big)$.  Taking the log on both sides gives the desired condition in the theorem statement.
\end{proof}

\begin{figure*}
    \begin{center}
        \includegraphics[width=0.9\textwidth]{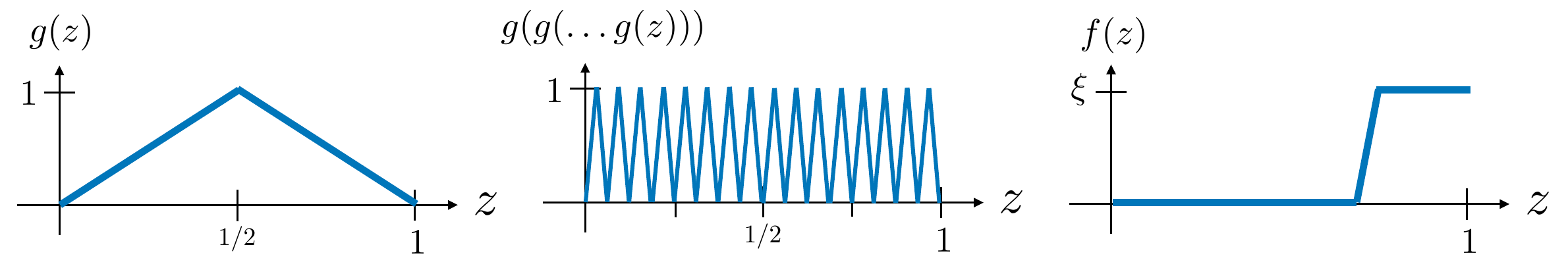}
    \end{center}
    \caption{Depiction of a sawtooth function with one triangular region (Left) and $R = 16$ triangular regions (Middle), along with a function that can be composed with the sawtooth function to produce one or more trapezoidal shapes (Right).} \label{fig:sawtooth} 
\end{figure*}

The final statement of Theorem \ref{thm:wide_deep} reveals that the upper and lower bounds are matching or near-matching:
\begin{itemize}
    \item For the depth-$2$ high-width network, under the mild assumption $\log k \le O\big( k_0 \log \frac{n_0}{k_0} \big)$, the sample complexity is $\Theta\big(k k_0 \log \frac{n_0}{k_0}\big)$ (i.e., the upper and lower bounds match).
    \item For the width-$O(n)$ high-depth network, the sample complexity is between $\Omega\big( k_0 \log \frac{n_0}{k_0} \big)$ and $O\big( k k_0 \log \frac{n_0}{k_0} \cdot \log n \big)$, with the two matching up to a $\log n$ factor.
    \item When both the depth and width are large, the lower bound is again $\Omega\big( k_0 \log \frac{n_0}{k_0} \big)$, and the upper bound has matching behavior whenever $k_0 \log \frac{n_0}{k_0} = \Omega( d \log n )$, while matching up to at most a $\log n$ factor more generally.
\end{itemize}
These findings reveal that the dependencies in Corollary \ref{cor:bora3} are optimal or near-optimal in the absence of further assumptions.

\begin{remark} \label{rem:forall2}
    In the same way as Remark \ref{rem:forall}, our analysis implicitly proves a stronger lower bound that applies even with respect to the average-case performance over some ``hard'' distribution on $\bx^*$ not depending on $\bA$.
\end{remark}

\section{Conclusion} \label{sec:conclusion}

We have established, to our knowledge, the first lower bounds on the sample complexity for compressive sensing with generative models.  To achieve these, we constructed generative models capable of producing group-sparse signals, and then applied minimax lower bounding techniques for group-sparse recovery.  For bounded Lipschitz-continuous generative models we matched the $O(k \log L)$ scaling law derived in \cite{Bor17}, and for ReLU-based generative models, we showed that the dependence of the $O(kd \log w)$ bound from \cite{Bor17} has an optimal or near-optimal dependence on both the depth and width.  A possible direction for future research is to understand what additional assumptions could be placed on the generative model to further reduce the sample complexity.
 

\appendix

\subsection{Covering Numbers} \label{sec:covering}

Let $(\calX,d)$ be a metric space, and fix $\epsilon>0$. A subset $S \subseteq \calX$ is said be an {\em $\epsilon$-net} of $\calX$ if, for all $x \in \calX$, there exists some $s \in S$ such that $d(s,x) \le \epsilon$.
The minimal cardinality of an $\epsilon$-net of $\calX$ (assuming it is finite) is denoted by $\calN^*(\calX,\epsilon)$ and is called the {\em covering number} of $\calX$ (with parameter $\epsilon$).  The following lemma is standard, but a short proof is included for completeness.

\begin{lemma}\label{lem:covering_number}
    {\em (Bound on the Covering Number)}
    The cube $\calX=B_{\infty}^k(r) \in \bbR^k$ equipped with the Euclidean metric satisfies for every $\epsilon > 0$ that 
    \begin{equation}
    \calN^*(\calX,\epsilon)  \le \left(1+\frac{2\sqrt{k}r}{\epsilon}\right)^k. 
    \end{equation}
\end{lemma}
\begin{proof}
    Let $\calN^*_\epsilon$ be a maximal $\epsilon$-separated subset of $\calX$. In other words, $\calN^*_\epsilon$ is such that $d(\bx,\by)\ge \epsilon$ for all $\bx,\by \in \calN^*_\epsilon$, $\bx \ne \by$, and no additional points can be added to the set while maintaining $\epsilon$-separation.  The maximality property implies that $\calN^*_\epsilon$ is an $\epsilon$-net of $\calX$, and the separation property implies that the balls of radius $\frac{\epsilon}{2}$ centered at the points in $\calN^*_\epsilon$ are disjoint. On the other hand, for any $\bx \in \calX$, we have $\|\bx\|_\infty \le r$ and thus $\|\bx\|_2 \le \sqrt{k} r$. Hence, all the balls of radius $\frac{\epsilon}{2}$ centered at the points in $\calN^*_\epsilon$ lie within the ball of radius $\big(\sqrt{k} r + \frac{\epsilon}{2}\big)$ centered at zero. Hence, comparing the volumes (which are proportional to the $k$-th power of the radius), we deduce that
    \begin{equation}
    |\calN^*_\epsilon| \le \left(1+\frac{2\sqrt{k}r}{\epsilon}\right)^k,
    \end{equation}
    which implies that $\calN^*(\calX,\epsilon)  \le \big(1+\frac{2\sqrt{k}r}{\epsilon}\big)^k$ as desired. 
\end{proof}

%
%

\subsection{Proof of Lemma \ref{lem:group} (Minimax Lower Bound for Group-Sparse Recovery)}  \label{sec:pf_sparse}

The lower bound for $k$-group-sparse recovery is proved using similar steps to those of regular $k$-sparse recovery.  For the latter setting, the original proof in \cite{Can13} used a standard reduction to multiple hypothesis testing, along with an application of Fano's inequality \cite[Sec.~2.10]{Cov01}.  We instead follow a more recent proof \cite{Duc13} (see also \cite{scarlett2019introductory}) based on a reduction to {\em approximate recovery} in multiple hypothesis testing, as this circumvents the need for an application of the non-elementary matrix Bernstein inequality.

We define the following set of $k$-group-sparse vectors with $\pm 1$ non-zero entries:
\begin{multline}
    \calV = \big\{ \bv \in \{-1,0,1\}^n \,:\, \bv \text{ is } k \text{-group-sparse} \\ \text{with exactly $k$ non-zero entries} \big\}.
\end{multline}
To each $\bv \in \calV$, we associate a vector $\bx_\bv = \xi \bv$;  we will later set $\xi = \sqrt{\frac{n \sigma^2 \log\frac{n}{k}}{4\|\bA\|_\rmF^2}}$ as per the lemma statement.  Letting $d_{\rm H}(\bv,\bv') = \sum_{i=1}^n \boldsymbol{1} \{ v_i \ne v'_i \}$ denote the Hamming distance, we have the following properties:
\begin{itemize}
    \item For each $\bv,\bv' \in \calV$, if $d_{\rm H}(\bv,\bv') > t$, then $\|\bx_\bv - \bx_{\bv'}\|_2 > \xi \sqrt{t}$;
    \item The cardinality of $\calV$ is $|\calV| = 2^k \big(\frac{n}{k}\big)^k$, yielding $\log|\calV| = k \log \frac{2n}{k}$;
    \item Letting $N_{\max}(t)$ be the maximum possible number of $\bv' \in \calV$ such that $d_{\rm H}(\bv,\bv') \le t$ for a fixed $\bv$, a simple counting argument gives $N_{\max}\big( \frac{k}{2} \big) \le \sum_{j=0}^{k/2} 2^j {n \choose j} \le  k \cdot 2^{k/2} \cdot {n \choose k/2}$, which gives 
    \begin{equation}
    \log N_{\max}\bigg( \frac{k}{2} \bigg) \le \log k + \frac{k \log 2}{2} + \frac{k}{2} \log \frac{2en}{k}.
    \end{equation}
    When $n \ge C_0 k$ for some absolute constant $C_0$, this implies that
    \begin{equation}
    \log \frac{|\calV|}{N_{\max}(\frac{k}{2})} \ge \frac{k}{3} \log \frac{n}{k}. \label{eq:V_ratio}
    \end{equation}
\end{itemize}

With these definitions, we can apply the general-purpose minimax lower bound from \cite{Duc13} (based on a form of Fano's inequality with approximate recovery), which reads as follows in our setting.  Here, for compactness, we write $\calM_m(k, \bA,\bar{\calS}_k(\xi)) = \inf_{\hat{\bx}} \sup_{\bx^* \in \bar{\calS}_k(\xi)} \bbE_{\bx^*} \left[\|\bx^*-\hat{\bx}\|_2^2\right]$ for the minimax risk. 

\begin{lemma}
    {\em (Estimation Lower Bound via Fano's Inequality \cite[Cor.~2]{Duc13})}
    Consider the preceding minimax setup with a fixed (deterministic) choice of $\bA$, and let $\calV$ be any set such that for any $\bv,\bv' \in \calV$, we have
    \begin{equation}
        d_{\rm H}(\bv,\bv') > t ~~\implies~~ \| \bx_{\bv} - \bx_{\bv'} \|_2 \ge \epsilon
    \end{equation}
    for some $t > 0$ and $\epsilon > 0$.  Then, the following holds:
    \begin{equation}
        \calM_m(k, \bA,\bar{\calS}_k(\xi)) \ge \Big( \frac{\epsilon}{2} \Big)^2 \bigg( 1 - \frac{I(V;\by) + \log 2}{ \log\frac{|\calV|}{N_{\max}(t)}} \bigg),
    \end{equation}
    where the mutual information is with respect to $V \to \bx_{V} \to \by$ with $V$ uniform on $\calV$ and $\by = \bA\bx_{V} + \bmeta$.
\end{lemma}

Substituting the above-established values $t = \frac{k}{2}$, $\epsilon = \xi \sqrt{ \frac{k}{2} }$, and \eqref{eq:V_ratio} into this result, we obtain
\begin{equation}
\calM_m (k, \bA,\bar{\calS}_k(\xi)) \ge \frac{k \xi^2}{8} \bigg( 1 - \frac{I(V;\by) + \log 2}{\frac{k}{3}\log\frac{n}{k}} \bigg).
\end{equation}
Using the assumption of Gaussian noise, it is shown in~\cite[Eq.~(87)]{scarlett2019introductory} (based on a similar result in \cite[p.~7]{Duc13}) that
\begin{equation}
I(V;\by) \le \frac{\xi^2}{2\sigma^2} \bbE\big[ \| \bA V \|_2^2 \big]. \label{eq:mi0}
\end{equation}
We claim that, as is the case for regular $k$-sparse signals in \cite{Duc13,scarlett2019introductory}, we have $\mathrm{Cov}[V] = \frac{k}{n} \bI_n$. Assuming this to be true momentarily, we observe that \eqref{eq:mi0} gives $I(V;\by) \le \frac{\xi^2}{2\sigma^2} \cdot \frac{k}{n} \|\bA\|_\rmF^2$, which yields 
\begin{equation}
    \calM_m (k, \bA,\bar{\calS}_k(\xi)) \ge \frac{k\xi^2}{8} \bigg( 1 - \frac{\frac{\xi^2}{2\sigma^2} \cdot \frac{k}{n} \|\bA\|_\rmF^2 + \log 2}{\frac{k}{3}\log\frac{n}{k}} \bigg).
\end{equation}
Setting $\xi = \sqrt{ \frac{n \sigma^2 \log \frac{n}{k}}{4 \|\bA\|_\rmF^2 } }$, we find that the bracketed term is lower bounded by $\frac{1}{2}$ when $n \ge C_0 k$ for an absolute constant $C_0$,\footnote{To see this, note that the bracketed term would {\em exactly} equal $\frac{1}{2}$ if the $\log 2$ term were dropped and we were to replace $\frac{k}{3}$ by $\frac{k}{4}$ in the denominator.} and this yields \eqref{eq:minimax_sigma}.  

To check that $\mathrm{Cov}[V] = \frac{k}{n} \bI_n$, due to symmetry we only need to check a single diagonal term and two cross-terms -- one cross-term for $V_i$ and $V_j$ in the same group of size $\frac{n}{k}$ in the group-sparse signal $V = (V_1,\dotsc,V_n)$, and another cross-term for $V_i$ and $V_j$ in different groups.  

Firstly, note that $\bbE[V_i] = 0$ for each $i=1,\dotsc,n$, since each entry is equally likely to be $+1$ or $-1$.  For the diagonal, since $V_1 \in \{-1,0,1\}$, we have $\bbE[V_1^2] = \bbP[V_1 \in \{-1,1\}] = \frac{k}{n}$, since each of the indices in $\{1,\dotsc,\frac{n}{k}\}$ is equally likely to be the non-zero one. For the first cross-term, note that if $i$ and $j$ are in the same group then $V_i V_j = 0$ with probability one, since each group has only one non-zero entry.  On the other hand, if $i$ and $j$ are in different groups then $V_i$ and $V_j$ are independent (since the non-zero entry of each group can be viewed as being selected independently of all other groups), so $\bbE[V_i V_j] = \bbE[V_i] \bbE[V_j] = 0$.

\subsection{Definitions for ReLU Networks}
\label{sec:relu_defs}

The ReLU activation function is defined as $\sigma(z) = \max\{z,0\}$. A neural network generative model $G: \bbR^{k} \rightarrow  \bbR^n$ with $d$ layers can be written as 
\begin{equation}
 G(\bz) = \phi_d\left(\phi_{d-1}\left(\cdots \phi_2( \phi_1(\bz,\btheta_1), \btheta_2)\cdots, \btheta_{d-1}\right), \btheta_d\right),
\end{equation}
where $\phi_l(\cdot)$ is the functional mapping corresponding to the $l$-th layer, and $\btheta_l = (\bW_l,\bb_l)$ is the parameter pair for the $l$-th layer:  $\bW_l \in \bbR^{n_l \times n_{l-1}}$ is the matrix of weights, and $\bb_l \in \bbR^{n_l}$ is the vector of offsets, where $n_l$ is the number of neurons in the $l$-th layer.  Note that $n_0 = k$ and $n_d = n$. Defining $\bz^0 = \bz$ and $\bz^l= \phi_l(\bz^{l-1},\btheta_l)$, if we set $\phi_l(\bz^{l-1},\btheta_l) = \sigma(\bW_l \bz^{l-1}+ \bb_l)$ for $l = 1,2,\ldots,d$ (where the ReLU operation $\sigma(\cdot)$ is applied element-wise), then the network is called a ReLU network.  The number of layers $d$ is also known as the depth, and the maximum layer size $w = \max_{l} n_l$ is known as the width.  In this paper, we only require fairly simple results about the representation power of ReLU networks from \cite{Tel15,Tel16,Aro18}, though the literature on this topic extends far beyond these works (e.g., see \cite{Per19,Cha20} and the references therein).

\section*{Acknowledgment}

We thank Eric Price for helpful discussions on the notions of ``for-each'' vs.~``for-all'' recovery guarantees, which we address in Remarks \ref{rem:forall} and \ref{rem:forall2}.  We also thank Ioannis Panageas for informing us of the construction of sawtooth functions via ReLU networks when both the depth and width are large, which we use in the proof of Theorem \ref{thm:wide_deep}.

\bibliographystyle{IEEEtran}
\bibliography{techReports,JS_References}

\begin{IEEEbiographynophoto}{Zhaoqiang Liu} was born in China in 1991. He is currently a research fellow in School of Computing at the National University of Singapore (NUS). He received the B.Sc.\ degree in Mathematics from the Department of Mathematical Sciences at Tsinghua University (THU) in 2013 and the Ph.D.\ degree in Mathematics from the Department of Mathematics at NUS in 2017. His research interests are in machine learning, including unsupervised learning such as matrix factorization and deep learning. 
 \end{IEEEbiographynophoto}

\begin{IEEEbiographynophoto}{Jonathan Scarlett}
     (S'14 -- M'15) received 
     the B.Eng. degree in electrical engineering and the B.Sci. degree in 
     computer science from the University of Melbourne, Australia. 
     From October 2011 to August 2014, he
     was a Ph.D. student in the Signal Processing and Communications Group
     at the University of Cambridge, United Kingdom. From September 2014 to
     September 2017, he was post-doctoral researcher with the Laboratory for
     Information and Inference Systems at the \'Ecole Polytechnique F\'ed\'erale
     de Lausanne, Switzerland. Since January 2018, he has been an assistant
     professor in the Department of Computer Science and Department of Mathematics,
     National University of Singapore. His research interests are in
     the areas of information theory, machine learning, signal processing, and
     high-dimensional statistics. He received the Singapore National Research Foundation (NRF) fellowship, and the NUS Early Career Research Award.
 \end{IEEEbiographynophoto}

\end{document}

%% file: preamble.tex
\usepackage[mathscr]{eucal}
\usepackage{epsfig,epsf,psfrag}
\usepackage{amssymb,amsmath,amsfonts,latexsym}
\usepackage{amsmath,graphicx,bm,xcolor,url}
\usepackage[caption=false]{subfig} 
\usepackage{fixltx2e}
\usepackage{array}
\usepackage{verbatim}
\usepackage{bm}
\usepackage{algpseudocode}
\usepackage{algorithm}
\usepackage{verbatim}
\usepackage{textcomp}
\usepackage{mathrsfs}
\usepackage{epstopdf}
\usepackage{relsize}
\usepackage{cleveref} 
\usepackage{subfig}
 \usepackage{amsthm}

\catcode`~=11 \def\UrlSpecials{\do\~{\kern -.15em\lower .7ex\hbox{~}\kern .04em}} \catcode`~=13 

\allowdisplaybreaks[3]

\newcommand{\calB}{\mathcal{B}}

\newcommand{\calD}{\mathcal{D}}

\newcommand{\calM}{\mathcal{M}}
\newcommand{\calN}{\mathcal{N}}

\newcommand{\calS}{\mathcal{S}}

\newcommand{\calV}{\mathcal{V}}

\newcommand{\calX}{\mathcal{X}}

\newcommand{\bA}{\mathbf{A}}
\newcommand{\bb}{\mathbf{b}}

\newcommand{\bI}{\mathbf{I}}

\newcommand{\bv}{\mathbf{v}}

\newcommand{\bW}{\mathbf{W}}
\newcommand{\bx}{\mathbf{x}}

\newcommand{\by}{\mathbf{y}}

\newcommand{\bz}{\mathbf{z}}

\newcommand{\rmF}{\mathrm{F}}

\newcommand{\bbE}{\mathbb{E}}

\newcommand{\bbP}{\mathbb{P}}

\newcommand{\bbR}{\mathbb{R}}

\DeclareMathAlphabet{\mathbsf}{OT1}{cmss}{bx}{n}
\DeclareMathAlphabet{\mathssf}{OT1}{cmss}{m}{sl}

\DeclareSymbolFont{bsfletters}{OT1}{cmss}{bx}{n}  
\DeclareSymbolFont{ssfletters}{OT1}{cmss}{m}{n}
\DeclareMathSymbol{\bsfGamma}{0}{bsfletters}{'000}
\DeclareMathSymbol{\ssfGamma}{0}{ssfletters}{'000}
\DeclareMathSymbol{\bsfDelta}{0}{bsfletters}{'001}
\DeclareMathSymbol{\ssfDelta}{0}{ssfletters}{'001}
\DeclareMathSymbol{\bsfTheta}{0}{bsfletters}{'002}
\DeclareMathSymbol{\ssfTheta}{0}{ssfletters}{'002}
\DeclareMathSymbol{\bsfLambda}{0}{bsfletters}{'003}
\DeclareMathSymbol{\ssfLambda}{0}{ssfletters}{'003}
\DeclareMathSymbol{\bsfXi}{0}{bsfletters}{'004}
\DeclareMathSymbol{\ssfXi}{0}{ssfletters}{'004}
\DeclareMathSymbol{\bsfPi}{0}{bsfletters}{'005}
\DeclareMathSymbol{\ssfPi}{0}{ssfletters}{'005}
\DeclareMathSymbol{\bsfSigma}{0}{bsfletters}{'006}
\DeclareMathSymbol{\ssfSigma}{0}{ssfletters}{'006}
\DeclareMathSymbol{\bsfUpsilon}{0}{bsfletters}{'007}
\DeclareMathSymbol{\ssfUpsilon}{0}{ssfletters}{'007}
\DeclareMathSymbol{\bsfPhi}{0}{bsfletters}{'010}
\DeclareMathSymbol{\ssfPhi}{0}{ssfletters}{'010}
\DeclareMathSymbol{\bsfPsi}{0}{bsfletters}{'011}
\DeclareMathSymbol{\ssfPsi}{0}{ssfletters}{'011}
\DeclareMathSymbol{\bsfOmega}{0}{bsfletters}{'012}
\DeclareMathSymbol{\ssfOmega}{0}{ssfletters}{'012}

\newcommand{\btheta}{\bm{\theta}}

\newcommand{\bmeta}{\bm{\eta}}

\DeclareMathOperator*{\argmin}{arg\,min}

\theoremstyle{plain}
\newtheorem{theorem}{Theorem} 
\newtheorem{lemma}{Lemma}

\newtheorem{corollary}{Corollary}

\newtheorem{remark}{Remark}

\newcommand{\qednew}{\nobreak \ifvmode \relax \else
      \ifdim\lastskip<1.5em \hskip-\lastskip
      \hskip1.5em plus0em minus0.5em \fi \nobreak
      \vrule height0.75em width0.5em depth0.25em\fi}

